\documentclass[conference, a4paper]{IEEEtran}
\IEEEoverridecommandlockouts
\usepackage{subfigure}
\usepackage[]{graphicx, amsfonts, amsmath, amssymb, amstext, latexsym, float, color}

\newcommand{\be}{\begin{equation}}
\newcommand{\ee}{\end{equation}}
\newcommand{\ber}{\begin{eqnarray}}
\newcommand{\eer}{\end{eqnarray}}

\pdfoutput=1
\newtheorem{theorem}{Theorem}

\newtheorem{conjecture}[theorem]{Conjecture}

\newtheorem{definition}[theorem]{Definition}

\newtheorem{proposition}[theorem]{Proposition}
\newtheorem{remark}[theorem]{Remark}

\newenvironment{proof}[1][Proof]{\noindent\textbf{#1.} }{\ \rule{0.5em}{0.5em}}

\begin{document}

\title{Thinning, photonic beamsplitting, and a general discrete entropy power inequality}

\author{
  \IEEEauthorblockN{Saikat Guha}
  \IEEEauthorblockA{Raytheon BBN Technologies\\
    Email: sguha@bbn.com} 
  \and
  \IEEEauthorblockN{Jeffrey H. Shapiro}
  \IEEEauthorblockA{Massachusetts Institute of Technology\\
    Email: jhs@mit.edu}
    \and
  \IEEEauthorblockN{Ra{\' u}l Garc{\' i}a-Patr{\' o}n S{\' a}nchez}
  \IEEEauthorblockA{Universit{\' e} Libre de Bruxelles\\
    Email: rgarciap@ulb.ac.be}
    }

\maketitle
\begin{abstract}
Many partially-successful attempts have been made to find the most natural discrete-variable version of Shannon's entropy power inequality (EPI). We develop an axiomatic framework from which we deduce the natural form of a discrete-variable EPI and an associated entropic monotonicity in a discrete-variable central limit theorem. In this discrete EPI, the geometric distribution, which has the maximum entropy among all discrete distributions with a given mean, assumes a role analogous to the Gaussian distribution in Shannon's EPI. The entropy power of $X$ is defined as the mean of a geometric random variable with entropy $H(X)$. The crux of our construction is a discrete-variable version of Lieb's scaled addition $X \boxplus_\eta Y$ of two discrete random variables $X$ and $Y$ with $\eta \in (0, 1)$. We discuss the relationship of our discrete EPI with recent work of Yu and Johnson who developed an EPI for a restricted class of random variables that have ultra-log-concave (ULC) distributions. Even though we leave open the proof of the aforesaid natural form of the discrete EPI, we show that this discrete EPI holds true for variables with arbitrary discrete distributions when the entropy power is redefined as $e^{H(X)}$ in analogy with the continuous version. Finally, we show that our conjectured discrete EPI is a special case of the yet-unproven Entropy Photon-number Inequality (EPnI), which assumes a role analogous to Shannon's EPI in capacity proofs for Gaussian bosonic (quantum) channels.
\end{abstract}

\vspace{-10pt}
\section{Introduction}\label{sec:intro}
\vspace{-2pt}
The Entropy Power Inequality (EPI) for statistically independent real {\em continuous}-valued random variables $X$ and $Y$,
\begin{equation}
v(X+Y) \ge v(X) + v(Y)
\label{eq:EPI}
\end{equation}
was first stated by Shannon~\cite{Sha48}, where $v(X) = e^{2h(X)}/(2\pi e)$ is the {\em entropy power} of $X$, with $h(X) = -\int_{-\infty}^{\infty} p_X(x)\log p_X(x) dx$ being the differential entropy of $X$. Equality in~\eqref{eq:EPI} holds if and only if $X$ and $Y$ are Gaussian. Using ${\cal E}_G(t) \equiv \frac12\log\left(2 \pi e t\right)$ the entropy of a Gaussian random variable with variance $t$, we get $v(X) = {\cal E}_G^{-1}\left(h(X)\right)$. In other words, the entropy power of $X$ is the variance of a Gaussian random variable that has the same differential entropy as that of $X$. The EPI was subsequently proved by Stam~\cite{Sta59} and Blachman~\cite{Bla65}. The EPI~\eqref{eq:EPI} has proven to be a powerful tool in information theory. It has found use in converse proofs of various Gaussian channel coding theorems involving continuous-valued sources and channels, especially in the respective converse proofs. Some prominent examples are: the capacity region of the scalar Gaussian broadcast channel~\cite{Ber74}, its generalization to the MIMO case~\cite{Wei06, Moh06}, the private capacity of the Gaussian wiretap channel~\cite{Che78}, capacity region of the Gaussian interference channel~\cite{Cos85}, Gaussian source multiple-description problem~\cite{Zam99, Oza80}, and rate distortion region for multi-terminal Gaussian source coding~\cite{Ooh05}.

Due to this success of Shannon's EPI, it is natural to speculate if there is an EPI for statistically-independent discrete-valued random variables $X$ and $Y$ defined on the set of natural numbers ${\mathbb N}_0 \equiv \left\{0, 1, \ldots\right\}$, which has analogous implications as above in discrete-valued information and coding theorems. 

There have been several attempts to provide an EPI for discrete random variables. It is simple to see that~\eqref{eq:EPI} does not hold as is when $X$ and $Y$ are arbitrary discrete random variables and the differential entropy $h(X)$ is replaced by the discrete Shannon entropy $H(X)$; a simple counterexample occurring for when $X$ and $Y$ are deterministic. EPIs with Bernoulli sources were proven where the ``$+$" operation in~\eqref{eq:EPI} was taken modulo $2$~\cite{Wyn73,Wit74,Ahl77,Sha90}. Harremo{\" e}s and Vignat retained the use of ``$+$" as integer addition, and proved that~\eqref{eq:EPI} holds when $X$ and $Y$ are independent binomial ${\rm Bin}(n_X,1/2)$ and ${\rm Bin}(n_Y,1/2)$ variables, upon redefining $v(X)$ in~\eqref{eq:EPI} as $e^{2H(X)}/(2\pi e)$, with $H(X)$ being the discrete entropy~\cite{Har03}. A discrete EPI was proven for arbitrary i.i.d. integer-valued random variables, albeit of a non-standard asymmetric form~\cite{Hag14}. Yu and Johnson developed a natural form of the discrete EPI but one that holds only for ultra-log concave (ULC) distributed random variables~\cite{Joh07, Yu09, Yu09a, Joh10}. Recently, Woo and Madiman showed that an inequality very similar to~\eqref{eq:EPI} holds for uniform distributions over finite subsets of the integers~\cite{Woo15}.

In this paper, we present an axiomatic framework to develop a natural generalization of the EPI for arbitrary discrete random variables and an associated central limit theorem (CLT). We begin with an overview of Shannon's EPI in Section~\ref{sec:review}. Section~\ref{sec:discreteEPI} provides an intuitive development of our discrete EPI. In Section~\ref{sec:quantum}, we provide proofs leveraging the algebra of quantum optics, and a close connection of our discrete EPI to a quantum EPI that naturally emerges in capacity proofs for transmitting classical information over bosonic channels~\cite{Guh08}.

\section{Review of the Entropy Power Inequality}\label{sec:review}

Using the following scaling identity for entropy power:
\begin{equation}
v(\sqrt{\beta}X) = \beta v(X),
\vspace{-1pt}
\label{eq:scalingEP}
\end{equation}
\vspace{-2pt}
it is straightforward to see that Ineq.~\eqref{eq:EPI} can be restated as:
\begin{equation}
v(X \boxplus_\eta Y) \ge \eta v(X) + (1-\eta)v(Y), \, {\text{where}}
\label{eq:EPIscaled}
\end{equation}
\begin{equation}
X \boxplus_\eta Y \equiv \sqrt{\eta}X+\sqrt{1-\eta}\,Y
\label{eq:boxplus_continuous}
\end{equation}
is a scaled addition operation. The `linear' form of the EPI, 
\begin{equation}
h(X \boxplus_\eta Y) \ge {\eta}h(X) + (1-{\eta})h(Y),
\label{eq:EPIlinear}
\end{equation}
i.e., the fact that differential entropy is concave with respect to normalized linear combinations, was first stated by Lieb~\cite{Lie78}, a rigorous proof of which, and the proof of the fact that~\eqref{eq:EPIlinear} is equivalent to~\eqref{eq:EPI}, was later provided by Dembo, Cover and Thomas~\cite{Dem91}. That~\eqref{eq:EPI}, hence~\eqref{eq:EPIscaled} implies~\eqref{eq:EPIlinear} follows simply from the application of ${\cal E}_G(\cdot)$ to both sides of~\eqref{eq:EPIscaled}, and using the concavity of the logarithm function. The converse implication requires the entropy of a scaled random variable, $h(aX) = h(X) + \log |a|$~\cite{Ver06}. Consider random variables $X^* = X/\sqrt{\eta}$ and $Y^* = Y/\sqrt{1 - \eta}$, with $\eta = v(X)/\left(v(X)+v(Y)\right)$. Using the expression for $h(aX)$ stated above, it readily follows that $h(X^*) = h(Y^*)$. We then have, using~\eqref{eq:EPIlinear},
\begin{eqnarray}
h(X+Y) &=& h\left(\sqrt{\eta}X^* + \sqrt{1-\eta}\, Y^*\right)\\
&\ge& \eta h(X^*) + (1-\eta) h(Y^*) = h(X^*).\label{eq:cor_rescaledEPI}
\end{eqnarray}
Applying ${\cal E}_G^{-1}$ to both sides, and using~\eqref{eq:scalingEP}, yields the EPI~\eqref{eq:EPI}.

Scaling plays an important role in many proofs of the EPI. The fact that the EPI can be stated in terms of scaled random variables as in inequality~\eqref{eq:cor_rescaledEPI} was implicit in Verd{\' u} and Guo's proof of the EPI~\cite{Ver06}, but Johnson and Yu later realized the significance of this form to construct an EPI for ULC discrete random variables, as we will describe in Section~\ref{sec:discreteEPI}~\cite{Joh10}. 

Artstein, Ball, Barthe and Naor proved a stronger form of the EPI~\eqref{eq:EPI} for sums of independent continuous random variables~\cite{Art04}, a special case of which was the first rigorous proof of {\em monotonicity} of the convergence of differential entropy in the CLT, i.e., for i.i.d. $X_i$, the entropy of the normalized sum $h\left(Y_n\right)$, with $Y_n = \sum_{i=1}^n X_i/{\sqrt{n}}$, is monotone increasing in $n$, and converges to the entropy of the Gaussian (which has the maximum entropy for a given variance) as $n \to \infty$. Note that $H\left(X_1 + X_2\right)/\sqrt{2} \ge H(X_1)$ for i.i.d. $X_1$ and $X_2$ follows from~\eqref{eq:EPIlinear}, repeated application of which shows that $h\left(Y_{2^k}\right)$ is nondecreasing in $k$. This cruder version of monotonicity is already sufficient to prove the CLT~\cite{Shi75}. Alternative proofs of monotonicity of $h(Y_n), \forall n \in {\mathbb N}_0$ were later given by Tulino and Verd{\' u}~\cite{Tul06}, and by Madiman and Barron~\cite{Mad07}.

\section{EPI for discrete random variables}\label{sec:discreteEPI}

\subsection{An axiomatic development of the discrete EPI}\label{sec:axioms}

The main ingredient needed for a natural discrete-variable generalization of Shannon's EPI, i.e.,~\eqref{eq:EPIscaled} and \eqref{eq:EPIlinear} is a:
\begin{itemize}
\item {\bf scaled addition}: An appropriate definition of $X \boxplus_\eta Y$ for $X$ and $Y$ both defined on ${\mathbb N}_0$, and $\eta \in (0, 1)$.
\end{itemize}
Further, the definition of $\boxplus_\eta$ should be extendable to a random vector ${\boldsymbol X} = (X_1, \ldots, X_n)$, i.e., $\boxplus_{\boldsymbol \eta} {\boldsymbol X}$, where ${\boldsymbol \eta} \equiv (\eta_1, \ldots, \eta_n)$ with $\eta_i \ge 0$ and $\sum_{i=1}^n\eta_i = 1$, such that: (1) it reduces to the bivariate case for $n=2$, viz., $\boxplus_{(\eta, 1-\eta)} (X_1, X_2) = X_1 \boxplus_\eta X_2$; (2) it is commutative in the sense that $\boxplus_{\boldsymbol \eta} {\boldsymbol X}$ is invariant under an arbitrary (yet identical) permutation of the entries of ${\boldsymbol \eta}$ and ${\boldsymbol X}$; and (3) it is well behaved under a CLT:
\begin{itemize}
\item {\bf {\em limiting} distribution}: A distribution $p_{L, \lambda}[k]$, $k \in {\mathbb N}_0$ should exist that can be defined solely as a function of its mean $\lambda$, and is the limiting distribution in a CLT under $\boxplus_{\boldsymbol \eta}$ addition. In other words, $p_{L, \lambda}[k]$ is the distribution of $Y_n \equiv \boxplus_{(\frac{1}{n}, \ldots, \frac{1}{n})} {\boldsymbol X}$, as $n \to \infty$, for i.i.d. arbitrarily-distributed $\lambda$-mean random variables ${\boldsymbol X} \equiv \left\{X_i\right\}$.
\end{itemize}
For a complete analogy with Shannon's EPI, one would want $H(Y_n)$, the entropy of $Y_n$, to be monotonically increasing in $n \in \left\{1, 2, \ldots\right\}$. Thus, $p_{L,\lambda}$ should be the distribution with the maximum entropy for a given mean $\lambda$. We already know that the geometric distribution $p_{L,\lambda}[k] = \left(1+\lambda\right)^{-1}\left(\lambda/(1+\lambda)\right)^k$, $k \in {\mathbb N}_0$, has this property. So, we would like the $\boxplus_{\boldsymbol \eta}$ operation for which the above CLT holds with the geometric distribution being the limiting distribution. One would then define:
\begin{itemize}
\item {\bf entropy power}: $V(X)$ of $X$ as the mean of a random variable with distribution $p_{L, \lambda}$ that has entropy $H(X)$.
\end{itemize}
Note that in order for the above to make sense, the entropy of the limiting distribution $p_{L, \lambda}$ should be monotonic (increasing) in its mean $\lambda$. With the above, the following should hold:
\begin{itemize}
\item {\bf entropy power inequality}: The EPI for discrete random variables, i.e.,~\eqref{eq:EPIscaled} and \eqref{eq:EPIlinear} should hold with the Shannon entropy power $v(X)$ replaced by the discrete entropy power $V(X)$ and the differential entropy $h(X)$ replaced by the discrete Shannon entropy $H(X)$. Further, equality in both aforesaid forms of the discrete EPI should hold when $X$ and $Y$ both are distributed according to the limiting distribution $p_{L, \lambda}$ (possibly with different means).
\end{itemize}

The above discussion suggests that once we have the `correct' definition of the $\boxplus_{\boldsymbol \eta}$ operation for discrete random variables (i.e., one that is well behaved under the CLT with the limiting distribution being geometric), one would immediately obtain the natural discrete generalization of Shannon's EPI.

In the above framework for discrete-valued random variables, we chose to peg the definitions to the mean as opposed to the variance (as in the case of continuous random variables). In the discrete case, the law of small numbers (see footnote~\ref{footnote:LOSN}) and the corresponding maximum entropy property both require `thinning' the mean, whereas in the continuous case, the central limit theorem requires the thinning of the variance, which is achieved by multiplication by $\sqrt{\eta}$. Many have realized that the R{\' e}nyi thinning operation $T_\eta$ is the natural discrete-variable equivalent of multiplication by $\sqrt{\eta}$, and has the desirable effect of thinning the mean by a factor $\eta$~\cite{Joh07, Yu09, Yu09a, Joh10, Har07}.

\begin{definition}
Given $\eta \in (0, 1)$, and $X \in {\mathbb N}_0$, the random variable $Y \in {\mathbb N}_0$, obtained by $\eta$-thinning of $X$ (denoted, $Y = T_\eta X$), has the distribution of the sum $\sum_{n=1}^X Z_n$, where $Z_i$ are binary $\left\{0, 1\right\}$ valued i.i.d. Bernoulli($\eta, 1-\eta$) random variables that are independent of $X$. The p.m.f. of $Y$ is:
$p_Y[n] = \sum_{k = n}^\infty p_X[k] \binom{k}{n}\eta^n(1-\eta)^{k-n}$.
\end{definition}


\subsection{Yu and Johnson's prior work on the discrete EPI}

Yu and Johnson developed a promising line of approach to the discrete EPI for ultra-log concave (ULC) random variables (i.e., those with p.m.f.s $p_X[n], n \in {\mathbb N}_0$, for which $np_X[n]/p_X[n-1]$ is decreasing as $n$ increases). They defined the scaled addition
\begin{equation}
X \boxplus_\eta Y = T_\eta X + T_{1-\eta} Y,
\label{eq:boxplus_YJ}
\end{equation}
which extends naturally to multiple variables and is well behaved under the CLT with the associated limiting distribution being Poisson. For i.i.d. mean-$\lambda$ $\left\{X_i\right\}$ with p.m.f. $p_X[n]$, $Y_n \equiv \boxplus_{(\frac{1}{n}, \ldots, \frac{1}{n})} {\boldsymbol X} = \sum_{i=1}^n T_{1/n}X_i = T_{1/n}\left(\sum_{i=1}^nX_i\right)$ converges to the Poisson distribution of mean $\lambda$ as $n \to \infty$~\footnote{\label{footnote:LOSN}The distribution of $\sum_{i=1}^nX_i$ is the $n$-fold convolution of $p_X[n]$. When $X_i \sim$ Bernoulli($p$), $\sum_{i=1}^nX_i \sim$ Binomial($n, p$) and $T_{1/n}\left(\sum X_i\right) \sim$ Binomial($n, p/n$) $\to$ Poisson($p$) as $n \to \infty$. This special case is the classical Binomial-to-Poisson convergence, known as the ``law of small numbers"~\cite{Har07}.}. Since the limiting distribution is not geometric, the $\boxplus_\eta$ operation in~\eqref{eq:boxplus_YJ} does not satisfy the criteria in Section~\ref{sec:axioms}. However, within the class of all ULC random variables of mean $\lambda$, the Poisson($\lambda$) random variable maximizes the entropy~\cite{Joh07}. Furthermore, $H(Y_n)$ increases monotonically in $n$ if $p_X[n]$ is ULC~\cite{Yu09}. Motivated by this, Yu and Johnson defined entropy power $V_p(X) = {\cal E}_p^{-1}\left(H(X)\right)$ in terms of the entropy ${\cal E}_p(\lambda)$ of a Poisson($\lambda$) random variable, with the hope that the straightforward equivalents~\eqref{eq:EPI},~\eqref{eq:EPIscaled} and~\eqref{eq:EPIlinear} would hold, with $X$ and $Y$ restricted to ULC random variables. They proved that the linear form (concavity of entropy)~\eqref{eq:EPIlinear} holds~\cite{Yu09}, i.e.,
\begin{equation}
H\left(X \boxplus_\eta Y\right) \ge \eta H(X) + (1-\eta)H(Y),
\label{eq:EPIULClinear}
\end{equation}
for all ULC independent $X$ and $Y$. This was the first major step towards a discrete EPI. The equivalents of \eqref{eq:EPI} and \eqref{eq:EPIscaled},
\begin{eqnarray}
V_p(X + Y) &\ge& V_p(X)+V_p(Y), \, {\text{and}}\\
V_p(X \boxplus_\eta Y) &\ge& \eta V_p(X)+(1-\eta)V_p(Y)
\label{eq:EPIULCscaled}
\end{eqnarray}
were naturally conjectured~\cite{Yu09}, but were shown later {\em not} to hold in general, even for ULC $X$ and $Y$~\cite{Joh10}.

The key step in going from the EPI~\eqref{eq:EPI} to the linear form~\eqref{eq:EPIlinear} was the scaling identity for the entropy power~\eqref{eq:scalingEP}. If such an identity were to hold for any ULC variable $X$, i.e., $V_p\left(T_\eta X\right) = \eta V_p(X)$, then~\eqref{eq:EPIULClinear} would imply the conjectured discrete EPI~\eqref{eq:EPIULCscaled}. But since the conjecture was found to be false, the scaling identity above cannot be true. However, the following one-sided version of it was proved for ULC $X$:
\begin{equation}
V_p\left(T_\eta X\right) \ge \eta V_p(X), \,\, \eta \in (0, 1).
\label{eq:RTEPI_ULC}
\end{equation}
which was termed the restricted discrete EPI~\cite{Joh10}. Even though ~\eqref{eq:EPIULCscaled} does not hold in general, the restated form~\eqref{eq:cor_rescaledEPI} does hold for discrete ULC variables~\cite{Joh10}: Given ULC independent $X$ and $Y$, if $\exists$ $X^*$ and $Y^*$ s.t. $X=T_\eta X^*$ and $Y=T_{1-\eta} Y^*$ for some $\eta \in (0, 1)$, s.t. $H(X^*) = H(Y^*)$, then
\begin{equation}
H(X+Y) \ge H(X^*),
\label{eq:rescaledEPIULC}
\end{equation}
with equality iff $X$ and $Y$ are Poisson. The reason why~\eqref{eq:rescaledEPIULC} holds but not~\eqref{eq:EPIULCscaled}, is that finding $\eta$, $X^*$ and $Y^*$ that satisfy the aforesaid conditions is not always possible. This is unlike the continuous version~\eqref{eq:cor_rescaledEPI}, for which $\eta$, $X^*$ and $Y^*$ can always be constructed that satisfy the conditions, as shown in Section~\ref{sec:review}.

\subsection{The natural discrete generalization of Shannon's EPI}\label{sec:discreteEPI_GSG}

\begin{figure}
\centering
\includegraphics[width=0.9\columnwidth]{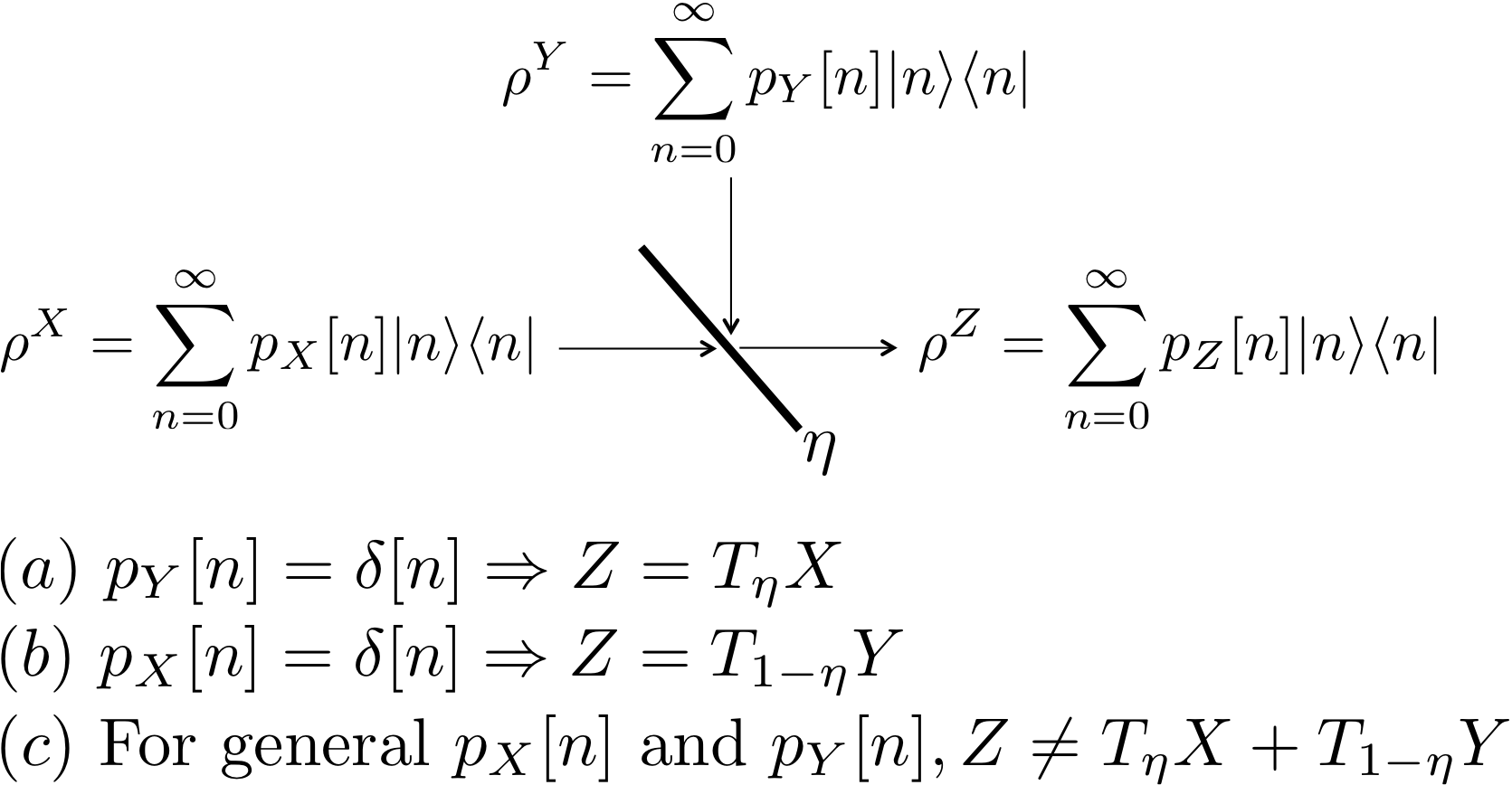}
\caption{Two independent photon-number-diagonal states $\rho^X$ and $\rho^Y$ combined on a beamsplitter with transmissivity $\eta$ results in a number-diagonal state $\rho^Z$ at the output. The distribution of the photon number at the output is a particular scaled addition of the two input distributions, explained in Section~\ref{sec:discreteEPI_GSG}, which is at the core of our discrete EPI.}
\label{fig:beamsplitter}
\end{figure}
Our goal is to develop a $\boxplus_\eta$ operation that: (i) is well behaved under a CLT with the geometric distribution as the limiting distribution; and (ii) satisfies Ineqs.~\eqref{eq:EPIULClinear} and~\eqref{eq:EPIULCscaled} for $X$ and $Y$ with {\em arbitrary} discrete distributions, and with the entropy power $V_g(X) = {\cal E}_g^{-1}\left(H(X)\right)$ defined in terms of ${\cal E}_g(\lambda) = (1+\lambda)\log(1+\lambda) - \lambda \log(\lambda)$, the entropy of the geometric distribution with mean $\lambda$.

Let us consider the following $1$-to-$1$ transform between a discrete distribution $p_X[n]$, $n \in {\mathbb N}_0$ and a circularly-symmetric continuous distribution $p_{{\boldsymbol X}_c}(\boldsymbol r)$, ${\boldsymbol r} \in {\mathbb C}$ (proof in Section~\ref{sec:quantum}):
\begin{eqnarray}
p_{{\boldsymbol X}_c}(\boldsymbol r) &=& \frac{1}{\pi}\sum_{n=0}^{\infty} p_X[n] \frac{e^{-\left| {\boldsymbol r} \right|^2} \left| {\boldsymbol r} \right|^{2n}}{n!}, {\text{and}}\label{eq:T}\\
p_X[n] &=& \frac{1}{\pi}\int p_{{\boldsymbol X}_c}(\boldsymbol r) {\cal L}_n\left(\left| {\boldsymbol s} \right|^2\right) e^{{\boldsymbol r}{\boldsymbol s}^* - {\boldsymbol r}^*{\boldsymbol s}}d^2{\boldsymbol r}d^2{\boldsymbol s},
\label{eq:inverseT}
\end{eqnarray}
where ${\cal L}_n(x)$ is the $n$-th Laguerre polynomial. Let us denote ${\boldsymbol X}_c = {\cal T}(X)$ and $X = {\cal T}^{-1}\left({\boldsymbol X}_c\right)$~\footnote{We are abusing notation a bit: ${\cal T}$ is not a function that maps the random variable $X$ to the random variable ${\boldsymbol X}_c$. It is a transform that takes a {\em function}, a discrete p.m.f. $p_X[n]$, to another function, a continuous p.d.f. $p_{{\boldsymbol X}_c}(\boldsymbol r)$.}. We define scaled addition $\boxplus_\eta$ for discrete random variables $X$ and $Y$ as follows\footnote{We will continue to use the same notation for the scaled addition, $\boxplus_\eta$, but in Section~\ref{sec:discreteEPI_GSG} onwards, the definition in Eq.~\eqref{eq:boxplus_GSG} will be implied for $\boxplus_\eta$.}:
\begin{equation}
X \boxplus_\eta Y \equiv {\cal T}^{-1}\left({\cal T}(X) \boxplus_\eta {\cal T}(Y)\right),
\label{eq:boxplus_GSG}
\end{equation}
where the $\boxplus_\eta$ on the RHS of~\eqref{eq:boxplus_GSG} is the usual continuous-variable definition as in~\eqref{eq:boxplus_continuous}. The following is simple to verify:
\begin{proposition}
If $p_Y[n] = \delta[n]$ is a delta function at $n=0$, $X \boxplus_\eta Y = T_\eta X$. Similarly, if $p_X[n] = \delta[n]$, $X \boxplus_\eta Y = T_{1-\eta} Y$.
\label{thm:thinning_pureloss}
\end{proposition}
\vspace{-5pt}
Our definition of $\boxplus_\eta$ coincides with the Yu-Johnson definition when one of the two variables is $0$ with probability $1$, but instead of directly adding the integer-valued scaled random variables $T_\eta X$ and $T_{1-\eta}Y$, we add the real-valued scaled random variables $\sqrt{\eta}{\boldsymbol X}_c$ and $\sqrt{1-\eta}{\boldsymbol Y}_c$ (and transform back to the integer domain using~\eqref{eq:inverseT}), where ${\boldsymbol X}_c = {\cal T}(X)$ and ${\boldsymbol Y}_c = {\cal T}(Y)$. Generalizing our $\boxplus_\eta$ operation for multiple variables is straightforward. Following the steps outlined in Section~\ref{sec:axioms}, we consider a CLT under this scaled addition.
\begin{theorem}
For i.i.d. arbitrarily-distributed $\lambda$-mean random variables ${\boldsymbol X} \equiv \left\{X_i\right\}$, the p.m.f. of $Y_n \equiv \boxplus_{(\frac{1}{n}, \ldots, \frac{1}{n})} {\boldsymbol X}$ converges to the geometric distribution of mean $\lambda$ as $n \to \infty$.
\label{thm:CLT_GSG}
\end{theorem}
\vspace{-10pt}
We have the following conjecture on monotonicity of entropy:
\begin{conjecture}
$H(Y_n)$ increases monotonically with $n \in \left\{1, 2, \ldots\right\}$ and converges to ${\cal E}_g(\lambda)$, as $n \to \infty$.
\label{conj:monotonicity}
\end{conjecture}
This is analogous to Harremo{\" e}s {\em et al.}'s law of small numbers~\cite{Har07} and Johnson-Yu's discrete entropic-monotonicity result~\cite{Joh10}, but no longer restricted to ULC random variables.
\begin{theorem}\label{thm:EPI_GSG_linear}
The linear form of the discrete EPI holds for arbitrary independent discrete-valued random variables $X$, $Y$:
\begin{equation}
H(X \boxplus_\eta Y) \ge \eta H(X) + (1-\eta) H(Y).
\label{eq:EPI_GSG_linear}
\end{equation}
\end{theorem}
We will show in Section~\ref{sec:quantum} that~\eqref{eq:EPI_GSG_linear} follows as a special case of a recent result on a quantum version of the EPI~\cite{Koe14}. Inequality~\eqref{eq:EPI_GSG_linear} with $\eta = 1/2$ implies $H\left(Y_{2^{k+1}}\right) \ge H\left(Y_{2^{k}}\right)$, which is sufficient to prove Theorem~\ref{thm:CLT_GSG} but only proves Conjecture~\ref{conj:monotonicity} for $n$ increasing in power-of-$2$ steps.
\begin{conjecture}\label{conj:EPI_GSG_scaled}
The following is the natural discrete generalization of Shannon's EPI, which holds true for arbitrary independent discrete-valued random variables $X$ and $Y$:
\begin{equation}
V_g(X \boxplus_\eta Y) \ge \eta V_g(X) + (1-\eta) V_g(Y).
\label{eq:EPI_GSG_scaled}
\end{equation}
\end{conjecture}
Even though we do not provide a proof of~\eqref{eq:EPI_GSG_scaled}, we will show in Section~\ref{sec:quantum} that ~\eqref{eq:EPI_GSG_scaled} is a simple special case of the yet-unproven entropy photon number inequality (EPnI)~\cite{Guh08, Guh07}.

De Palma {\em et al.} recently proved the following restricted one-sided version of~\eqref{eq:EPI_GSG_scaled}, analogous to Yu and Johnson's Ineq.~\eqref{eq:RTEPI_ULC}:
\begin{equation}
V_g\left(T_\eta X\right) \ge \eta V_g(X), \,\, \eta \in (0, 1).
\label{eq:RTEPI_Palma}
\end{equation}
See Theorem $23$ of Ref.~\cite{Pal16} for a proof. There has been a suite of recent progress in quantum versions of the EPI~\cite{Guh08, Guh07, Guha_PhDthesis_08, Gio04, Gio14, Koe14, Pal15, Pal14, Pal15a, Pal16}, and an eventual proof of the EPnI will imply the validity of the discrete EPI~\eqref{eq:EPI_GSG_scaled} in Conjecture~\ref{conj:EPI_GSG_scaled}. Finally, we have:
\begin{theorem}
Using a (somewhat unnatural) definition of entropy power, $V_e(X) = e^{H(X)}$---in analogy with the continuous entropy power $v(X)$---the EPI statement in~\eqref{eq:EPI_GSG_scaled} holds true for arbitrary independent discrete random variables $X$, $Y$:
\begin{equation}
V_e(X \boxplus_\eta Y) \ge \eta \, V_e(X) + (1 - \eta) \, V_e(Y).
\label{eq:EPIdiscrete}
\end{equation}
\end{theorem}
We show in Section~\ref{sec:quantum} that this result follows as a special case of a quantum EPI result recently proved in~\cite{Pal15, Pal14}.

\section{Proofs of discrete EPI results}\label{sec:quantum}

In this Section, we will provide proofs of the various statements in Section~\ref{sec:discreteEPI_GSG}. Even though it is possible to prove them directly, it is much easier to leverage the mathematics of quantum optics. Let us consider a beamsplitter of transmissivity $\eta \in (0, 1)$ that mixes two modes whose annihilation operators are $\hat x$ and $\hat y$, to produce an output mode ${\hat z} = \sqrt{\eta}\,{\hat x} + \sqrt{1-\eta}\,{\hat y}$. Assume that the quantum states of the two input modes, $\rho^{XY} = \rho^X \otimes \rho^Y$ are statistically independent. The density operators $\rho^X$ and $\rho^Y$ are infinite-dimensional, unit-trace, positive, Hermitian matrices, whose matrix elements we will express in the complete orthonormal {\em Fock} (or photon-number) basis $|n\rangle$, $n = 0, 1, \ldots, \infty$. The von Neumann entropy of $\rho^X$, $S(\rho^X) = -{\rm Tr}\left(\rho^X \log \rho^X\right) = H(\left\{\lambda^X_i\right\}) = -\sum_i \lambda_i^X\log \lambda_i^X$, where  $\left\{\lambda^X_i\right\}$ are the eigenvalues of $\rho^X$. For a number diagonal state $\rho^X = \sum_{n=0}^\infty p_X[n]|n\rangle \langle n|$, $S(\rho^X) = H(X)$, Shannon entropy of the discrete random variable $X$. 
\begin{theorem}
Consider independent $X$ and $Y$ with p.m.f.s $p_X[n]$ and $p_Y[n]$, $n \in {\mathbb N}_0$. Consider mixing independent number-diagonal states $\rho^X = \sum_{n=0}^\infty p_X[n]|n\rangle \langle n|$ and $\rho^Y = \sum_{n=0}^\infty p_Y[n]|n\rangle \langle n|$ on a beamsplitter of transmissivity $\eta$. Then the output state is also number-diagonal, i.e., $\rho^Z = \sum_{n=0}^\infty p_Z[n]|n\rangle \langle n|$, and the output number distribution $p_Z[n]$ is that of the random variable $Z = X \boxplus_\eta Y$, with our definition of $\boxplus_\eta$ given in~\eqref{eq:boxplus_GSG}. In other words, the scaled addition in Eq.~\eqref{eq:boxplus_GSG} has a physical interpretation---it is how a beamsplitter of transmissivity $\eta$ `adds' photon number distributions of two independent number-diagonal input states.
\end{theorem}

\begin{proof}
The Husimi function of a quantum state $\rho^X$, $p_{{\boldsymbol X}_c}(\boldsymbol r) = \frac{1}{\pi}\langle {\boldsymbol r}| \rho^X | {\boldsymbol r}\rangle$, ${\boldsymbol r} \in {\mathbb C}$, can be interpreted as the p.d.f. of a continuous-valued random variable ${{\boldsymbol X}_c}$. Here, $|{\boldsymbol r}\rangle$ is the {\em coherent state} of complex amplitude ${\boldsymbol r}$, an eigenstate of the annihilation operator ${\hat x}$. When $\rho^X = \sum_{n=0}^\infty p_X[n]|n\rangle \langle n|$ is number-diagonal, its Husimi function $p_{{\boldsymbol X}_c}(\boldsymbol r)$ is given by Eq.~\eqref{eq:T}, which is a circularly-symmetric function in the phase space. The anti-normally-ordered characteristic function of ${\rho}^X$ is then $\chi_A^X({\boldsymbol s}) = \int p_{{\boldsymbol X}_c}(\boldsymbol r) e^{{\boldsymbol r}{\boldsymbol s}^* - {\boldsymbol r}^*{\boldsymbol s}}d^2{\boldsymbol r}$. Using the operator Fourier inverse to express the state $\rho^X = \int \chi_A^X({\boldsymbol s}) e^{-{\boldsymbol s}{\hat x}^\dagger} e^{{\boldsymbol s}^*{\hat x}} d^2{\boldsymbol s}/\pi$, writing down the number-basis diagonal elements $\langle n | \rho^X |n \rangle = p_X[n]$, and using the identity $\langle n | e^{-{\boldsymbol s}{\hat x}^\dagger} e^{{\boldsymbol s}^*{\hat x}} | n \rangle = {\cal L}_n\left(\left| {\boldsymbol s} \right|^2\right)$, we obtain Eq.~\eqref{eq:inverseT}. So, we have the $1$-to-$1$ transform, ${\boldsymbol X}_c = {\cal T}(X)$ and $X = {\cal T}^{-1}\left({\boldsymbol X}_c\right)$. Physically, $X$ and ${\boldsymbol X_c}$ are random variables corresponding to the outcomes of (ideal) photon number measurement and (ideal) optical heterodyne detection measurement respectively, on the state $\rho^X$. Because the Husimi function of a quantum state is unique, the above transform relation implies that the Husimi function of a state is circularly symmetric if and only if it is diagonal in the number basis.

Next, we observe that the Husimi function of the output state $\rho^Z$ is given by the scaled convolution $p_{{\boldsymbol Z}_c}(\boldsymbol r) = \frac{1}{\eta}\,p_{{\boldsymbol X}_c}\left(\frac{\boldsymbol r}{\sqrt{\eta}}\right) \ast \frac{1}{(1-\eta)}\,p_{{\boldsymbol Y}_c}\left(\frac{\boldsymbol r}{\sqrt{1-\eta}}\right)$, which implies that the respective transformed variables are related by the scaled addition, ${\boldsymbol Z}_c = \sqrt{\eta}\, {\boldsymbol X}_c + \sqrt{1-\eta}\, {\boldsymbol Y}_c \equiv {\boldsymbol X}_c \boxplus_\eta {\boldsymbol Y}_c$~\cite{Guha_PhDthesis_08}. Given $\rho^X$ and $\rho^Y$ are both number diagonal (and hence have circularly-symmetric Husimi functions), $\rho^Z$ must also have a circularly-symmetric Husimi function, and hence be number diagonal. Therefore, the number distribution of $\rho^Z$ is that of a random variable $Z = {\cal T}^{-1}\left({\cal T}(X) \boxplus_\eta {\cal T}(Y)\right)$.
\end{proof}

\begin{remark}
It is simple to verify that, for number-diagonal $\rho^X$, if $p_Y[n] = \delta[0]$, the output state's number distribution $p_Z[n]$ is the $\eta$-thinned version of $p_X[n]$, $Z = T_\eta X$. Similarly, with $p_X[n] = \delta[0]$ and number-diagonal $\rho^Y$, $Z = T_{1-\eta}Y$. Interestingly however, if neither $\rho^X$ nor $\rho^Y$ is the vacuum state, the output number distribution is {\em not} a simple addition of the thinned input distributions, i.e., $Z \ne T_\eta X + T_{1-\eta}Y$ (The RHS is Yu and Johnson's $\boxplus_\eta$); rather it is given by~\eqref{eq:boxplus_GSG}.
\end{remark}

Proof of the discrete CLT in Theorem~\ref{thm:CLT_GSG} follows from Theorem 5.10 of~\cite{Guha_PhDthesis_08} when restricted to number-diagonal inputs and using the definition of $\boxplus_\eta$ in Eq.~\eqref{eq:boxplus_GSG}.


Let us consider the following two definitions of entropy power for a quantum state, and associated entropy power inequalities for each definition: 

(a) {\em entropy photon number}, ${\cal V}_g(\rho^X) = {\cal E}_g^{-1}\left(S(\rho^X)\right)$, where ${\cal E}_g(\lambda) = (1+\lambda)\log(1+\lambda) - \lambda \log \lambda$ is the von Neumann entropy of a thermal state $\rho_{t,\lambda} = \sum_{n=0}^\infty p_{t,\lambda}[n]|n\rangle \langle n|$ of mean photon number $\lambda$, with $p_{t,\lambda}[n] = \left(1+\lambda\right)^{-1}\left(\lambda/(1+\lambda)\right)^n$. Since $\rho_{t,\lambda}$ is diagonal in the number basis, ${\cal E}_g(\lambda)$ is also the Shannon entropy of the geometric distribution of mean $\lambda$, as used in Section~\ref{sec:discreteEPI_GSG}.

(b) {\em quantum entropy power}, ${\cal V}_e(\rho^X) = {\cal E}_e^{-1}\left(S(\rho^X)\right) = e^{S\left(\rho^X\right)}$, where ${\cal E}_e(\lambda) = \log(\lambda)$, defined in analogy with the Gaussian entropy function ${\cal E}_G(t)$ used to define the classical entropy power $v(X)$ that appears in the original EPI~\eqref{eq:EPI}.

{\em Entropy photon number inequality} (EPnI)~\cite{Guh08}---For a pair of independent states $\rho^X$ and $\rho^Y$ input to a beamsplitter of transmissivity $\eta$, producing the state $\rho^Z$ at the output, Guha, Shapiro and Erkmen conjectured the following~\cite{Guh08}:
\begin{equation}
{\cal V}_g\left(\rho^Z\right) \ge \eta\,{\cal V}_g\left(\rho^X\right) + (1-\eta)\,{\cal V}_g\left(\rho^Y\right),
\label{eq:EPnI}
\end{equation}
proving that equality occurs if $\rho^X$ and $\rho^Y$ are thermal states. The EPnI plays a role analogous to the EPI in converse proofs of capacities of Gaussian bosonic channels. A general proof of EPnI remains open, but it was proved for Gaussian-state inputs~\cite{Guh08}. It was shown that a special case of the EPnI,
\begin{equation}
{\cal V}_g\left(T_\eta\rho^X\right) \ge \eta {\cal V}_g\left(\rho^X\right),
\label{eq:qRTEPI_MOE2}
\end{equation}
a statement analogous to~\eqref{eq:RTEPI_ULC}, would complete the converse proofs of the capacity region of the multi-user bosonic broadcast channel~\cite{Guh07, Guha_PhDthesis_08}, and the triple tradeoff region of the pure-loss bosonic channel~\cite{Wil12}. Inequality~\eqref{eq:qRTEPI_MOE2} was proved recently by De Palma {\em et al.}~\cite{Pal16}, which when applied to number-diagonal inputs, implies~\eqref{eq:RTEPI_Palma} as a special case.

{\em Quantum entropy power inequality} (q-EPI)~\cite{Koe14}---Koenig and Smith conjectured (and provided a partial proof for) the following, by replacing the ${\cal V}_g$ in the EPnI by ${\cal V}_e$~\cite{Koe14}:
\begin{equation}
{\cal V}_e\left(\rho^Z\right) \ge \eta\,{\cal V}_e\left(\rho^X\right) + (1-\eta)\,{\cal V}_e\left(\rho^Y\right).
\label{eq:qEPI}
\end{equation}
A complete proof of the q-EPI~\eqref{eq:qEPI}, and a multi-input generalization thereof, were provided by De Palma {\em et al.}~\cite{Pal15, Pal14}. However, unlike the EPnI, which emerges naturally in the bosonic Gaussian channel capacity converses, the q-EPI only implies upper bounds (or outer bounds, in case of broadcast channels) to the respective capacities~\cite{Koe13} (or capacity regions, for broadcast channels~\cite{Pal14}). The q-EPI implies the following linear form, analogous to~\eqref{eq:EPIlinear}, by applying ${\cal E}_e(\cdot)$ on both sides of~\eqref{eq:qEPI} and using the concavity of the log function:
\begin{equation}
S\left(\rho^Z\right) \ge \eta\,S\left(\rho^X\right) + (1-\eta)\,S\left(\rho^Y\right).
\label{eq:qEPIlinear}
\end{equation}
If the EPnI were proven true, the concavity of ${\cal E}_g(\cdot)$ would also imply~\eqref{eq:qEPIlinear}. 
The linear form~\eqref{eq:qEPIlinear} when applied to number-diagonal inputs implies~\eqref{eq:EPI_GSG_linear} in Theorem~\ref{thm:EPI_GSG_linear}. Similarly, the q-EPI~\eqref{eq:qEPI} implies the discrete EPI~\eqref{eq:EPIdiscrete}.
By employing~\eqref{eq:qEPIlinear} for $\eta = 1/2$ recursively on pairs of identical input states, we get:
$
S\left(\rho^{Y_{2^{k+1}}}\right) \ge S\left(\rho^{Y_{2^{k}}}\right), \, k = 0, 1, \ldots,
$
which provides a partial proof of a conjecture by Guha and Shapiro on entropic monotonicity in a quantum CLT (Theorem 5.10 of~\cite{Guha_PhDthesis_08}), with $n$ increasing in power-of-$2$ increments. A complete proof of the aforesaid conjecture, when applied to identical independent number-diagonal states, will imply Conjecture~\ref{conj:monotonicity} as a special case. Finally, a proof of the EPnI would imply the natural discrete-variable generalization of Shannon's EPI~\eqref{eq:EPI_GSG_scaled}, which would satisfy all the desirable properties stated in Section~\ref{sec:axioms} and would hold for all discrete random variables and not be restricted to ULC distributed random variables as in the Yu-Johnson discrete EPI.

\section{Conclusions}

Shannon's entropy power inequality (EPI) found many applications in proving coding theorem converses for many Gaussian channel and source coding problems. The Entropy Photon-number Inequality (EPnI) was shown to assume a role analogous to Shannon's EPI in capacity converse proofs for transmitting classical information over Gaussian bosonic (quantum) channels. Even though the general form of the EPnI remains unproven, several special cases of it have been proven in the recent years. 

Many attempts have been made to find the most natural discrete-variable version of Shannon's entropy power inequality (EPI). In this paper, we developed an axiomatic framework from which we deduced the natural form of a discrete-variable EPI and an associated entropic monotonicity in a discrete-variable central limit theorem. In this discrete EPI, the geometric distribution, which has the maximum entropy among all discrete distributions with a given mean, assumes a role analogous to the Gaussian distribution in Shannon's EPI, and the thermal state in the EPnI. We defined the entropy power of a discrete random variable $X$ as the mean of a geometric random variable with entropy $H(X)$. The crux of our construction is a discrete-variable version of Lieb's scaled addition $X \boxplus_\eta Y$ of two discrete random variables $X$ and $Y$ with $\eta \in (0, 1)$. We discussed the relationship of our discrete EPI with recent work of Yu and Johnson who developed an EPI for a restricted class of random variables that have ultra-log-concave (ULC) distributions, and pegged their definition of the entropy power to the entropy function of the Poisson distribution, which attains the maximum entropy for a given mean, among the class of ULC random variables. Even though we left open the proof of the aforesaid natural form of the discrete EPI that we conjectured in this paper, we showed that this discrete EPI holds true for discrete random variables with arbitrary discrete distributions when the entropy power is redefined as $e^{H(X)}$ in analogy with the continuous version. Finally, we showed that our conjectured discrete EPI is a special case of the EPnI, corresponding to the case when two input quantum states to the EPnI are independent number-diagonal states.

\section*{acknowledgements}
S.G. and R.G.-P. acknowledge the organizers of the third Beyond i.i.d. Workshop held at BIRS, Banff in July 2015, where this work took seed, and thank Rupert L. Frank and Elliott Lieb for inspiring discussions on the EPI at that workshop. The authors thank Vittorio Giovannetti and Seth Lloyd for useful discussions. SG would like to acknowledge the DARPA Quiness program funded under US
Army Research Office award number W31P4Q-12-1-0019. J.H.S. acknowledges support from the Air Force Office of
Scientific Research (AFOSR) grant number FA9550-14-1-0052. R.G.-P. is Research Associate of the F.R.S.-FNRS.

This document does not contain technology or technical data controlled under either the U.S. International Traffic in Arms Regulations or the U.S. Export Administration Regulations.

\end{document}